\documentclass[preprint, 12pt]{elsarticle}

\usepackage{amsmath,amsthm,amsfonts,amssymb, graphicx, subfigure}

\numberwithin{equation}{section}
\theoremstyle{plain}

\newtheorem{theorem}{Theorem}[section]

\newtheorem{corollary}[theorem]{Corollary}

\newtheorem{proposition}[theorem]{Proposition}

\newtheorem{remark}[theorem]{Remark}

\journal{Physica A}

\begin{document}

\begin{frontmatter}

\title{Reciprocity in directed networks}

\author[label1]{Mei Yin\corref{cor1}}

\address[label1]{Department of Mathematics, University of Denver, Denver,
CO 80208, USA}

\cortext[cor1]{Corresponding author. Tel.: +1 303 871 2130; fax:
+1 303 871 3173.}

\ead{mei.yin@du.edu}

\author[label2]{Lingjiong Zhu}

\address[label2]{School of Mathematics, University of Minnesota, Minneapolis, MN 55455, USA}

\ead{zhul@umn.edu}

\begin{abstract}
Reciprocity is an important characteristic of directed networks and has been widely used in the modeling
of World Wide Web, email, social, and other complex networks.
In this paper, we take a statistical physics point of view and study
the limiting entropy and free energy densities from the microcanonical ensemble, the canonical ensemble,
and the grand canonical ensemble
whose sufficient statistics are given by edge and reciprocal densities.
The sparse case is also studied for the grand canonical ensemble. Extensions to more general reciprocal models including reciprocal triangle and star densities will likewise be discussed.
\end{abstract}


\end{frontmatter}

\section{Introduction}
\label{intro}
Reciprocity evaluates the tendency of vertex pairs to form mutual connections between each other
and is an important object to study in complex networks, such as
email networks, see e.g. Newman et al. \cite{NewmanII},
World Wide Web, see e.g. Albert et al. \cite{Albert},
World Trade Web, see e.g. Gleditsch \cite{Gleditsch},
social networks, see e.g. Wasserman and Faust \cite{Wasserman}, and cellular networks, see e.g. Jeong et al. \cite{Jeong}.
In networks that aggregate temporal information, reciprocity
provides a measure of the simplest feed-back process occurring on the network,
i.e., the tendency of one stimulus, a vertex, to respond to another stimulus, another vertex. Reciprocity is important because most complex networks are directed and it is the main quantity characterizing feasible dyadic patterns, namely possible types of connections between two nodes.
One example is the email network.
Just because user B's email address appears in user A's address book does not necessarily mean that the reverse
is also true, although it often is, see e.g. Newman et al. \cite{NewmanII}.
Another example is the social network. Reciprocity captures a basic way in which different forms
of interaction take place on a social network like Twitter. When two users A and B interact
as peers, one expects that messages will be exchanged between them in both directions.
However, if user A sends messages to user B, who is a celebrity or news source, it is likely
that user B will not send messages in return, see e.g. Cheng et al. \cite{Cheng}.
Therefore, it is not enough to just understand the \emph{edge density} of a directed network,
the \emph{reciprocal density} needs to be studied as well.
In Garlaschelli and Loffredo \cite{Garlaschelli}, it was discovered that
detecting nontrivial patterns of reciprocity can reveal mechanisms and organizing principles that
help explain the topology of the observed network. They also proposed a measure of reciprocity and studied how strong it is for different complex networks, and found that reciprocity is strongest in the World Trade Web.
People often treat complex networks as undirected for simplicity, and reciprocity can help quantify the information loss induced by projecting a directed network into an undirected one. Using the knowledge of reciprocity,
significant directed information can be retrieved from an undirected projection, and the error introduced when a directed network is treated as undirected may be estimated, see e.g. Garlaschelli and Loffredo \cite{GarlaschelliII}.

Directed networks consisting of $n$ nodes can be modeled by directed graphs on $n$ vertices,
where a graph is represented by a matrix
$X = (X_{ij})_{1\le i,j \le n}$ with each $X_{ij} \in \{0,1\}$.
Here, $X_{ij} = 1$ means there is a directed edge from vertex $i$ to
vertex $j$; otherwise, $X_{ij}=0$. We assume that $(X_{ii})_{1\le i \le n}=0$ so that there are no self-loops. Give the set of such graphs the probability
\begin{equation}\label{1}
{\mathbb P}_n^{\beta_1,\beta_2}(X) = Z_n(\beta_{1},\beta_{2})^{-1}\exp\left[n^2\left(\beta_{1} e(X) + \beta_{2} r(X)\right)\right],
\end{equation}
where
\begin{equation}\label{es}
e(X) := n^{-2}\sum_{1\leq i,j\leq n} X_{ij}, \quad
r(X) := n^{-2}\sum_{1\leq i,j\leq n} X_{ij}X_{ji},
\end{equation}
$\beta_1$ and $\beta_2$ are parameters, and $Z_n(\beta_{1},\beta_{2})$ is the appropriate normalization.
Note that $e(X)$ and $r(X)$, defined in~\eqref{es}, respectively represent
the directed \emph{edge density} and the \emph{reciprocal density}.

In the literature, $X_{ij}$ and $X_{ij}X_{ji}$ are sometimes referred to as the \emph{single edge} and the \emph{reciprocal edge}.
This belongs to the class of exponential random graph models called $p_{1}$ models of Holland and Leinhardt \cite{Holland}.
Further extensions include $p_{2}$ models, see e.g. Lazega and van Duijn \cite{Lazega} and van Duijn et al. \cite{Duijn}.
More general types of exponential models have also been introduced and studied.
See Besag \cite{Besag}, Newman \cite{Newman}, Rinaldo et al. \cite{Rinaldo}, Robins et al. \cite{Robins},
Snijders et al. \cite{Snijders}, Wasserman and Faust \cite{Wasserman}, and Fienberg \cite{FienbergI, FienbergII} for history and a review of developments. The exponential random graph models have popular counterparts in statistical physics:
a hierarchy of models ranging from the \emph{grand canonical ensemble}, the \emph{canonical ensemble}, to the \emph{microcanonical ensemble},
with particle density and energy density in place of $e(X)$ and $r(X)$, and temperature and chemical potential
in place of $\beta_1$ and $\beta_2$. In the grand canonical ensemble, the reciprocal model (\ref{1}) in this case,
no prior knowledge of the graph is assumed. In the canonical ensemble, partial information of the graph is given.
For instance, the edge density of the graph is close to $1/2$ or the reciprocal density is close to $1/4$.
In the microcanonical ensemble, complete information of the graph is observed beforehand, say in the reciprocal model,
both the edge density and the reciprocal density are specified.

It is well-known that models in this hierarchy have a very simple relationship involving Legendre transforms and,
more importantly, the \emph{free energy density} (of the grand canonical ensemble),
the \emph{conditional free energy density} (of the canonical ensemble), and the \emph{entropy density}
(of the microcanonical ensemble) encode important information of a random graph drawn from the model. See illustration below. As one goes down the hierarchy, the model is understood from varying perspectives: the free energy and conditional free energy densities characterize the macroscopic and mesoscopic configurations of the system respectively, while the entropy density describes the degree to which the probability of the system is spread out over different possible microstates. Various objects of interest are obtained by differentiating these densities with respect to appropriate parameters and phases are determined by analyzing the singularities of the derivatives. In particular, they serve as a measure of how close the system is to equilibrium, namely perfect internal disorder, and their monotonicity sheds light on the relative likelihood of each configuration following the philosophy that the higher the entropy the greater the disorder. Since real-world networks are often very large in
size, the infinite-size asymptotics of these quantities have received exponentially growing attention in recent years.
See e.g. Aristoff and Zhu \cite{AristoffZhu, AristoffZhuII}, Chatterjee and Dembo \cite{ChatterjeeDembo},
Chatterjee and Diaconis \cite{ChatterjeeDiaconis}, Kenyon et al. \cite{Kenyon}, Kenyon and Yin \cite{KenyonYin},
Lubetzky and Zhao \cite{LubetzkyZhao, LubetzkyZhaoII}, Radin and Sadun \cite{RadinII, RadinIV}, Radin et al. \cite{RadinIII},
Radin and Yin \cite{Radin}, Yin \cite{Yin}, Yin et al. \cite{YinII}, and Zhu \cite{Zhu}. It may be worth pointing out that most of these papers utilize the theory of graph limits as developed by Lov\'{a}sz and coworkers \cite{Lovasz2009, Lov}.

\vskip.1truein

\begin{center}
The hierarchy
\begin{equation*}
\begin{array}{cc}
\text{grand canonical ensemble} & \text{free energy density}
\\
\downarrow & \downarrow
\\
\text{canonical ensemble} & \text{conditional free energy density}
\\
\downarrow & \downarrow
\\
\text{microcanonical ensemble} &  \text{entropy density}
\end{array}
\end{equation*}
\end{center}

\vskip.1truein

The rest of this paper is organized as follows. In Section \ref{microcanonical} we derive the exact expression
for the normalization constant (partition function) of the reciprocal model (the grand canonical ensemble) and
analyze the asymptotic features of its associated microcanonical ensemble. Our main results are: an exact expression
for the limiting entropy density (Theorem \ref{exact}) and some discussion on its monotonicity
(Remark \ref{mono1}). In Section \ref{canonical} we investigate the asymptotic features of two canonical ensembles associated with the reciprocal model,
one conditional on the edge density and the other conditional on the reciprocal density. Our main results are:
exact expressions for the two limiting conditional free energy densities (Theorem \ref{exact2}) and some discussion
on their monotonicity (Remark \ref{mono23}). In Section \ref{grand} we take another look
at the reciprocal model and examine its asymptotic features. Our main results are: a joint central limit theorem describing convergence of the edge density
and the reciprocal density (Proposition \ref{CLT}), exact scalings for the limiting normalization constant
(Theorem \ref{scaling}) and the mean of the limiting probability distribution in the sparse regime (Proposition \ref{sparsemean}).
Lastly, in Section \ref{discussion} we extend our analysis to more general reciprocal models whose sufficient statistics,
besides single edge and reciprocal edge, also include reciprocal $p$-star and reciprocal triangle.
Large deviations techniques are used throughout this paper. We refer the readers to the works of Chatterjee and Diaconis \cite{ChatterjeeDiaconis} and Chatterjee and Varadhan \cite{ChatterjeeVaradhan} for more details of this framework.

\section{The microcanonical ensemble}
\label{microcanonical}
After extracting the exponential factor in the reciprocal model (\ref{1}), each possible configuration of the directed graph is weighted equally. This amounts to taking $(X_{ij})_{1\leq i\neq j\leq n}$ as iid Bernoulli random variables having values $1$ and $0$
each with probability $1/2$. Denote the associated probability measure and the associated expectation by $\mathbb{P}_n$ and $\mathbb{E}_n$ respectively. Define
\begin{equation}
\lambda_{n,\delta}(\epsilon,r) = \frac{1}{n^2}\log {\mathbb P}_n\left(|e(X)-\epsilon|<\delta,\,
|r(X)-r|<\delta\right).
\end{equation}
Shrink the intervals around $\epsilon$ and $r$ by letting $\delta$ go to zero, we are interested in the limit
\begin{equation}\label{micropsi}
\lambda(\epsilon,r) := \lim_{\delta \to 0} \lim_{n\to \infty}  \lambda_{n,\delta}(\epsilon,r).
\end{equation}
The quantity in \eqref{micropsi} will be called the \emph{limiting entropy density}. Via the theory of large deviations, it is directly connected to the \textit{limiting free energy density}
\begin{equation}\label{psig}
\chi(\beta_1, \beta_2):=\lim_{n\rightarrow \infty}\frac{1}{n^2}\log Z_n(\beta_1, \beta_2).
\end{equation}

\begin{theorem}\label{Lyapunov}
\begin{equation}
\label{chi}
\chi(\beta_1, \beta_2)=\frac{1}{2}\log\left(1+2e^{\beta_{1}}+e^{2\beta_{1}+2\beta_{2}}\right).
\end{equation}
\end{theorem}

\begin{proof}[Proof of Theorem \ref{Lyapunov}]
Recall that by assumption, $(X_{ii})_{1\leq i\leq n}=0$. Thus
\begin{equation}
\sum_{1\leq i,j\leq n}X_{ij}=\sum_{1\leq i<j\leq n}X_{ij}+\sum_{1\leq j<i\leq n}X_{ij}=\sum_{1\leq i<j\leq n}(X_{ij}+X_{ji}),
\end{equation}
\begin{equation*}
\sum_{1\leq i,j\leq n}X_{ij}X_{ji}=\sum_{1\leq i<j\leq n}X_{ij}X_{ji}+\sum_{1\leq j<i\leq n}X_{ij}X_{ji}=2\sum_{1\leq i<j\leq n}X_{ij}X_{ji}.
\end{equation*}
This implies that
\begin{align}
Z_{n}(\beta_{1},\beta_{2})
&=2^{n(n-1)}\mathbb{E}_n\left[e^{\beta_{1}\sum_{1\leq i,j\leq n}X_{ij}+\beta_{2}\sum_{1\leq i,j\leq n}X_{ij}X_{ji}}\right]
\\
&=2^{n(n-1)}\mathbb{E}_n\left[e^{\beta_{1}\sum_{1\leq i<j\leq n}(X_{ij}+X_{ji})
+2\beta_{2}\sum_{1\leq i<j\leq n}X_{ij}X_{ji}}\right]
\nonumber
\\
&=2^{n(n-1)}\prod_{1\leq i<j\leq n}\mathbb{E}_n\left[e^{\beta_{1}(X_{ij}+X_{ji})
+2\beta_{2}X_{ij}X_{ji}}\right]
\nonumber
\\
&=\left(1+2e^{\beta_{1}}+e^{2\beta_{1}+2\beta_{2}}\right)^{\binom{n}{2}}.
\nonumber
\end{align}
Hence we draw the conclusion.
\end{proof}

\begin{corollary}
\label{cor}
\begin{equation}
\label{sup}
\lambda(\epsilon,r)=-\sup_{\beta_{1},\beta_{2}\in\mathbb{R}}
\left\{\beta_{1}\epsilon+\beta_{2}r
-\frac{1}{2}\log\left(\frac{1}{4}+\frac{1}{2}e^{\beta_{1}}+\frac{1}{4}e^{2\beta_{1}+2\beta_{2}}\right)\right\}.
\end{equation}
\end{corollary}

\begin{proof}[Proof of Corollary \ref{cor}]
From the proof of Theorem \ref{Lyapunov},
\begin{multline}
\lim_{n\rightarrow\infty}\frac{1}{n^{2}}\log \mathbb{E}_n\left[e^{\beta_{1}\sum_{1\leq i,j\leq n}X_{ij}+\beta_{2}\sum_{1\leq i,j\leq n}X_{ij}X_{ji}}\right]
\\=\frac{1}{2}\log\left(\frac{1}{4}+\frac{1}{2}e^{\beta_{1}}+\frac{1}{4}e^{2\beta_{1}+2\beta_{2}}\right),
\end{multline}
which up to a constant is essentially (\ref{chi}), and is finite for any $\beta_{1},\beta_{2}\in\mathbb{R}$ and differentiable in both $\beta_{1}$ and $\beta_{2}$.
The result then follows from G\"{a}rtner-Ellis theorem in large deviations theory, see e.g. Dembo and Zeitouni \cite{Dembo}, which states that the entropy $\lambda(\epsilon,r)$ may be obtained as the Legendre transform of the free energy $\chi(\beta_1,\beta_2)$.
\end{proof}

\begin{remark}
(i) Note that $0\leq \frac{1}{n^2}\sum_{1\leq i,j\leq n}X_{ij}\leq 1$ and\\ $0\leq \frac{1}{n^2}\sum_{1\leq i,j\leq n}X_{ij}X_{ji}\leq 1$, which implies that $\lambda(\epsilon,r)=-\infty$ if $\epsilon\notin[0,1]$ or $r\notin[0,1]$.

(ii) Note that $\sum_{1\leq i,j\leq n}X_{ij}X_{ji}\leq\sum_{1\leq i,j\leq n}X_{ij}$, which implies that $\lambda(\epsilon,r)=-\infty$
if $r>\epsilon$.

(iii) Note that
\begin{align}
\sum_{1\leq i,j\leq n}(X_{ij}X_{ji}+1)-2\sum_{1\leq i,j\leq n}X_{ij}
\nonumber
&=\sum_{1\leq i,j\leq n}(X_{ij}X_{ji}+1-X_{ij}-X_{ji})
\\
\nonumber
&=\sum_{1\leq i,j\leq n}(X_{ij}-1)(X_{ji}-1)\geq 0,
\nonumber
\end{align}
which implies that $\lambda(\epsilon,r)=-\infty$ if $1+r-2\epsilon<0$.
\end{remark}

\begin{theorem}
\label{exact}
For $\epsilon,r\in[0,1]$, $\epsilon\geq r$ and $1+r-2\epsilon\geq 0$,
\begin{align}
\label{epsilonr}
\lambda(\epsilon,r)&=-\epsilon\log\left(\frac{\epsilon-r}{1+r-2\epsilon}\right)
-\frac{r}{2}\log\left(\frac{r(1+r-2\epsilon)}{(\epsilon-r)^{2}}\right)
\\
&\qquad\qquad\qquad
+\frac{1}{2}\log\left(\frac{1}{4(1+r-2\epsilon)}\right),
\nonumber
\end{align}
and otherwise $\lambda(\epsilon,r)=-\infty$.
\end{theorem}

\begin{proof}[Proof of Theorem \ref{exact}]
Under the assumption that $\epsilon,r\in [0,1]$, it is easy to see that the supremum in (\ref{sup}) can not be obtained at $\beta_1,\beta_2=\pm\infty$,
and $\lambda(\epsilon,r)$ must attain its extremum at finite $\beta_1,
\beta_2$. At optimality,
\begin{align}
&\epsilon=\frac{e^{\beta_{1}}+e^{2\beta_{1}+2\beta_{2}}}
{1+2e^{\beta_{1}}+e^{2\beta_{1}+2\beta_{2}}},\label{EqnI}
\\
&r=\frac{e^{2\beta_{1}+2\beta_{2}}}
{1+2e^{\beta_{1}}+e^{2\beta_{1}+2\beta_{2}}}.\label{EqnII}
\end{align}
Dividing \eqref{EqnI} by \eqref{EqnII}, we get
\begin{equation}
\frac{\epsilon}{r}=1+e^{-\beta_{1}-2\beta_{2}}.
\end{equation}
Substitute this back into \eqref{EqnI},
\begin{equation}
\epsilon=\frac{e^{-\beta_{1}-2\beta_{2}}+1}
{e^{-2\beta_{1}-2\beta_{2}}+2e^{-\beta_{1}-2\beta_{2}}+1}
=\frac{\frac{\epsilon}{r}}
{e^{-\beta_{1}}(\frac{\epsilon}{r}-1)+\frac{2\epsilon}{r}-1},
\end{equation}
which implies that
\begin{equation}
e^{\beta_{1}}=\frac{\epsilon-r}{1+r-2\epsilon},
\qquad
e^{2\beta_{2}}=\frac{r(1+r-2\epsilon)}{(\epsilon-r)^{2}}.
\end{equation}
The conclusion thus follows.
\end{proof}

\begin{remark}
It is straightforward to compute that
\begin{equation}
\lambda\left(\frac{1}{2},\frac{1}{4}\right)=-\frac{1}{2}\log\left(\frac{\frac{1}{4}}{\frac{1}{4}}\right)
-\frac{1}{8}\log\left(\frac{\frac{1}{4}\cdot\frac{1}{4}}{(\frac{1}{4})^{2}}\right)
+\frac{1}{2}\log\left(\frac{1}{4\cdot\frac{1}{4}}\right)=0.
\end{equation}
This is consistent with the law of large numbers and the maximal entropy principle.
\end{remark}

\begin{remark}
Along the Erd\H{o}s-R\'{e}nyi curve $r=\epsilon^{2}$, $0\leq\epsilon\leq 1$,
\begin{equation}
\label{rate}
\lambda(\epsilon,\epsilon^{2})=-\epsilon\log\epsilon-(1-\epsilon)\log(1-\epsilon)-\log 2.
\end{equation}
This is the entropy of a Bernoulli random variable and is minus the rate function
of the large deviations for the edge density.
\end{remark}

\begin{figure}
\centering
\begin{subfigure}
  \centering
  \includegraphics[height=1.785in, width=2.38in]{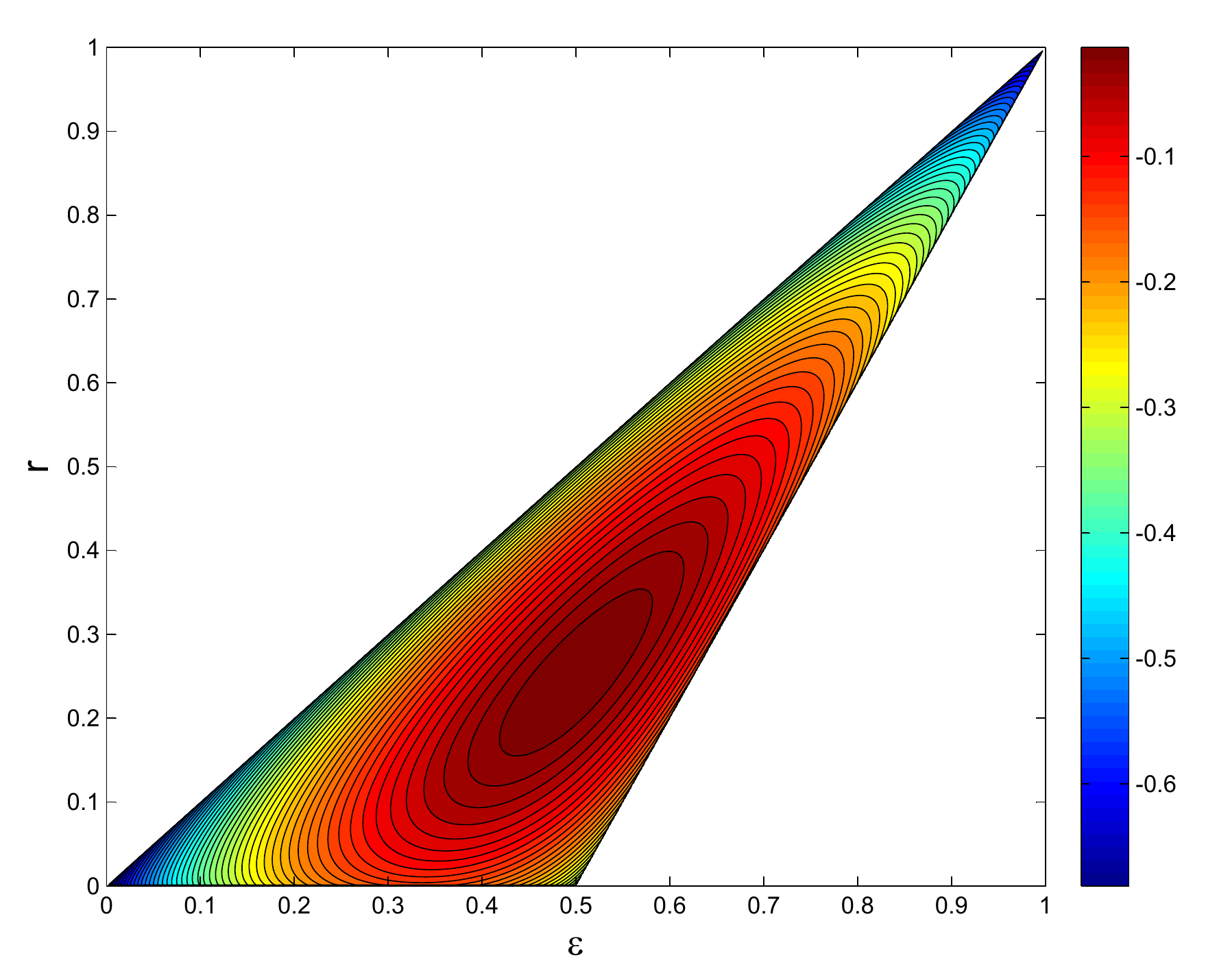}
\end{subfigure}
\begin{subfigure}
  \centering
  \includegraphics[height=1.785in, width=2.38in]{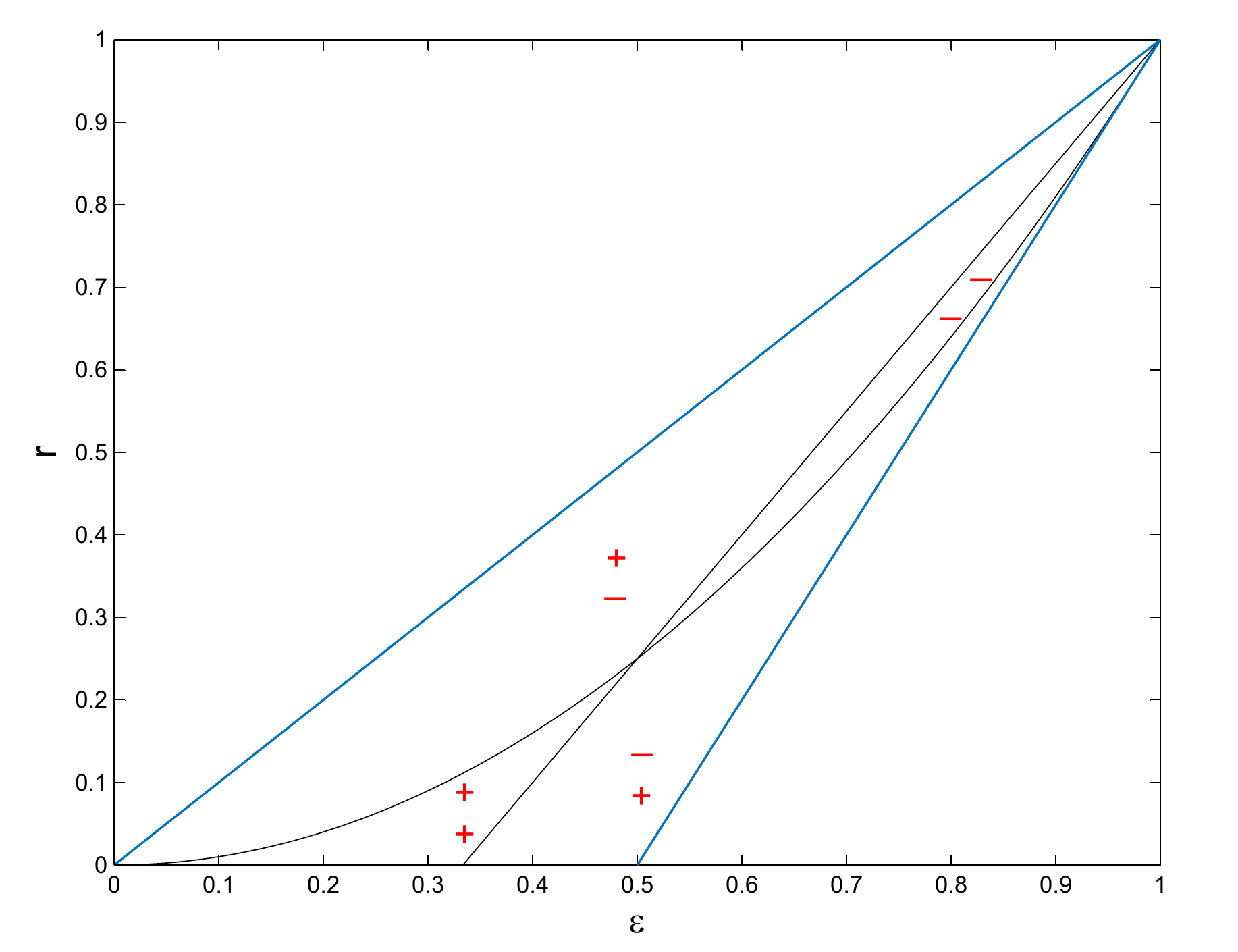}
\end{subfigure}
\caption{On the left hand side, we have the contour plot of the limiting entropy density $\lambda(\epsilon,r)$ obtained
from Theorem \ref{exact}. On the right hand side, we specify the regions of monotonicity as obtained in Remark \ref{mono1}. In region $^{-}_{-}$, $\lambda$ is decreasing
in both $\epsilon$ and $r$; in region $^{+}_{-}$, $\lambda$ is increasing in $\epsilon$ and decreasing
in $r$; in region $^{+}_{+}$, $\lambda$ is increasing in both $\epsilon$ and $r$; in region $^{-}_{+}$,
$\lambda$ is decreasing in $\epsilon$ and increasing in $r$. The boundaries are given by $1+2r=3\epsilon$
and $r=\epsilon^{2}$.}
\end{figure}

\begin{remark}
\label{mono1} Let us analyze the monotonicity of the limiting entropy density. On one hand,
\begin{equation}\label{DerivativeEpsilon}
\frac{\partial\lambda}{\partial\epsilon}
=-\log\left(\frac{\epsilon-r}{1+r-2\epsilon}\right),
\end{equation}
which implies that $\frac{\partial\lambda}{\partial\epsilon}\geq 0$
if and only if $\epsilon-r\leq 1+r-2\epsilon$, which is equivalent to $1+2r\geq 3\epsilon$. On the other hand,
\begin{equation}\label{DerivativeR}
\frac{\partial\lambda}{\partial r}
=-\frac{1}{2}\log\left(\frac{r(1+r-2\epsilon)}{(\epsilon-r)^{2}}\right),
\end{equation}
which implies that $\frac{\partial\lambda}{\partial r}\geq 0$ if and only if $r(1+r-2\epsilon)\leq(\epsilon-r)^{2}$,
which is equivalent to $r\leq\epsilon^{2}$, i.e., $\lambda(\epsilon,r)$ is increasing in $r$
below the Erd\H{o}s-R\'{e}nyi curve and decreasing in $r$
above the Erd\H{o}s-R\'{e}nyi curve. (See \cite{RadinIV} for a similar phenomenon across the
Erd\H{o}s-R\'{e}nyi curve in the (undirected) edge-triangle model.)
\end{remark}

\section{The canonical ensemble}
\label{canonical}
As in Aristoff and Zhu \cite{AristoffZhuII}, Kenyon and Yin \cite{KenyonYin} and Zhu \cite{Zhu}, we are interested in the asymptotic features of constrained models. The probability measure is given by
\begin{equation}
\mathbb{P}_{n,\delta}^{\epsilon,\beta_{2}}(X)
=\frac{1}{2^{n(n-1)}}\exp\left[n^{2}\left(\beta_2r(X)-\phi_{n,\delta}(\epsilon,\beta_{2})\right)\right]
1_{|e(X)-\epsilon|<\delta}
\end{equation}
if conditional on the edge density, and by
\begin{equation}
\mathbb{P}_{n,\delta}^{\beta_{1},r}(X)
=\frac{1}{2^{n(n-1)}}\exp\left[n^{2}\left(\beta_1e(X)-\psi_{n,\delta}(\beta_{1},r)\right)\right]
1_{|r(X)-r|<\delta}
\end{equation}
if conditional on the reciprocal density, where
\begin{align}
&\phi_{n,\delta}(\epsilon,\beta_{2})
=\frac{1}{n^{2}}\log \mathbb E_n\left[
\exp\left(n^2\beta_{2}r(X)\right)1_{|e(X)-\epsilon|<\delta}\right],
\\
&\psi_{n,\delta}(\beta_{1},r)
=\frac{1}{n^{2}}\log \mathbb E_n\left[
\exp\left(n^2\beta_{1}e(X)\right)1_{|r(X)-r|<\delta}\right].
\nonumber
\end{align}
We shrink the interval around $\epsilon$ (or $r$) by letting $\delta$ go to zero:
\begin{align}
\label{conditional}
&\phi(\epsilon,\beta_{2}):=\lim_{\delta\rightarrow 0}\lim_{n\rightarrow\infty}
\phi_{n,\delta}(\epsilon,\beta_{2}),
\\
&\psi(\beta_{1},r):=\lim_{\delta\rightarrow 0}\lim_{n\rightarrow\infty}
\psi_{n,\delta}(\beta_{1},r).
\nonumber
\end{align}
The quantities in (\ref{conditional}) will be called the \textit{limiting conditional free energy densities}.

\begin{theorem}
\label{exact2}
For any $\beta_{2}\in\mathbb{R}$, $0\leq\epsilon\leq 1$,
\begin{equation}
\phi(\epsilon,\beta_{2})=
-\epsilon\log\left(\frac{\epsilon-r^{\ast}}{1+r^{\ast}-2\epsilon}\right)
+\frac{1}{2}\log\left(\frac{1}{4(1+r^{\ast}-2\epsilon)}\right),
\end{equation}
where
\begin{equation}
r^{\ast}=
\begin{cases}
\frac{(2\epsilon e^{2\beta_{2}}-2\epsilon+1)-\sqrt{(2\epsilon e^{2\beta_{2}}-2\epsilon+1)^{2}-4\epsilon^{2}e^{2\beta_{2}}(e^{2\beta_{2}}-1)}}
{2(e^{2\beta_{2}}-1)} &\text{if $\beta_{2}\neq 0$},
\\
\epsilon^{2} &\text{if $\beta_{2}=0$}.
\end{cases}
\end{equation}
For any $\beta_{1}\in\mathbb{R}$, $0\leq r\leq 1$,
\begin{equation}
\psi(\beta_{1},r)=-\frac{r}{2}\log r-\log 2-\frac{1+r}{2}\log\left(\frac{1-r}{2e^{\beta_{1}}+1}\right)
+r\log\left(\frac{e^{\beta_{1}}(1-r)}{2e^{\beta_{1}}+1}\right).
\end{equation}
\end{theorem}

\begin{proof}[Proof of Theorem \ref{exact2}]
By using Varadhan's lemma, see e.g. Dembo and Zeitouni \cite{Dembo},
\begin{align}
&\phi(\epsilon,\beta_{2})=\sup_{2\epsilon-1\leq r\leq \epsilon}\{\beta_{2}r+\lambda(\epsilon,r)\},\label{FormulaI}
\\
&\psi(\beta_{1},r)=\sup_{r\leq \epsilon\leq \frac{r+1}{2}}\{\beta_{1}\epsilon+\lambda(\epsilon,r)\}.\label{FormulaII}
\end{align}
By \eqref{DerivativeR}, the optimal $r$ in \eqref{FormulaI} satisfies
\begin{equation}
0=\beta_{2}+\frac{\partial\lambda}{\partial r}=\beta_{2}-\frac{1}{2}\log\left(\frac{r(1+r-2\epsilon)}{(\epsilon-r)^{2}}\right),
\end{equation}
which is equivalent to
\begin{equation}
\label{TwoSolutions}
(e^{2\beta_{2}}-1)r^{2}-(2\epsilon e^{2\beta_{2}}-2\epsilon+1)r+\epsilon^{2}e^{2\beta_{2}}=0.
\end{equation}
When $\beta_{2}=0$, \eqref{TwoSolutions} has one solution $r^{\ast}=\epsilon^{2}$. When $\beta_{2}\neq 0$, since
\begin{align}
(2\epsilon e^{2\beta_{2}}-2\epsilon+1)^{2}-4\epsilon^{2}e^{2\beta_{2}}(e^{2\beta_{2}}-1)
&=4\epsilon^{2}+1-4\epsilon+4\epsilon e^{2\beta_{2}}-4\epsilon^2e^{2\beta_2}
\\
&=(2\epsilon-1)^{2}+4\epsilon(1-\epsilon)e^{2\beta_2}>0,
\nonumber
\end{align}
\eqref{TwoSolutions} has two solutions
\begin{equation}
r^{\pm}=\frac{(2\epsilon e^{2\beta_{2}}-2\epsilon+1)\pm\sqrt{(2\epsilon e^{2\beta_{2}}-2\epsilon+1)^{2}-4\epsilon^{2}e^{2\beta_{2}}(e^{2\beta_{2}}-1)}}
{2(e^{2\beta_{2}}-1)}.
\end{equation}
When $\beta_{2}<0$, one solution of \eqref{TwoSolutions}
is positive and the other is negative. We check that $r^{-}>0$ and $r^{+}<0$ and thus
the optimal $r^{\ast}=r^{-}$.
When $\beta_{2}>0$, both solutions of \eqref{TwoSolutions} are positive. We check that
\begin{equation}
r^{+}+r^{-}=\frac{2\epsilon e^{2\beta_{2}}-2\epsilon+1}
{e^{2\beta_{2}}-1}>2\epsilon
\end{equation}
and $r^{+}\geq\frac{r^{+}+r^{-}}{2}>\epsilon$ and thus the optimal $r^{\ast}=r^{-}$. This is indeed the optimizer following the mean value theorem, since $\frac{\partial\lambda}{\partial r}|_{r=\epsilon-}=-\infty$ and\\ $\frac{\partial\lambda}{\partial r}|_{r=(2\epsilon-1)+}=\infty$. By \eqref{DerivativeEpsilon}, the optimal $\epsilon$ in \eqref{FormulaII} satisfies
\begin{equation}
0=\beta_{1}+\frac{\partial\lambda}{\partial\epsilon}
=\beta_{1}-\log\left(\frac{\epsilon-r}{1+r-2\epsilon}\right),
\end{equation}
which has one solution
$\epsilon^{\ast}=\frac{e^{\beta_{1}}(1+r)+r}{2e^{\beta_{1}}+1}$. This is indeed
the optimizer following the mean value theorem, since $\frac{\partial\lambda}{\partial\epsilon}|_{\epsilon=r^{+}}=\infty$
and $\frac{\partial\lambda}{\partial\epsilon}|_{\epsilon=\left(\frac{1+r}{2}\right)^-}=-\infty$.
\end{proof}

\begin{figure}
\centering
\begin{subfigure}
  \centering
  \includegraphics[height=1.785in, width=2.38in]{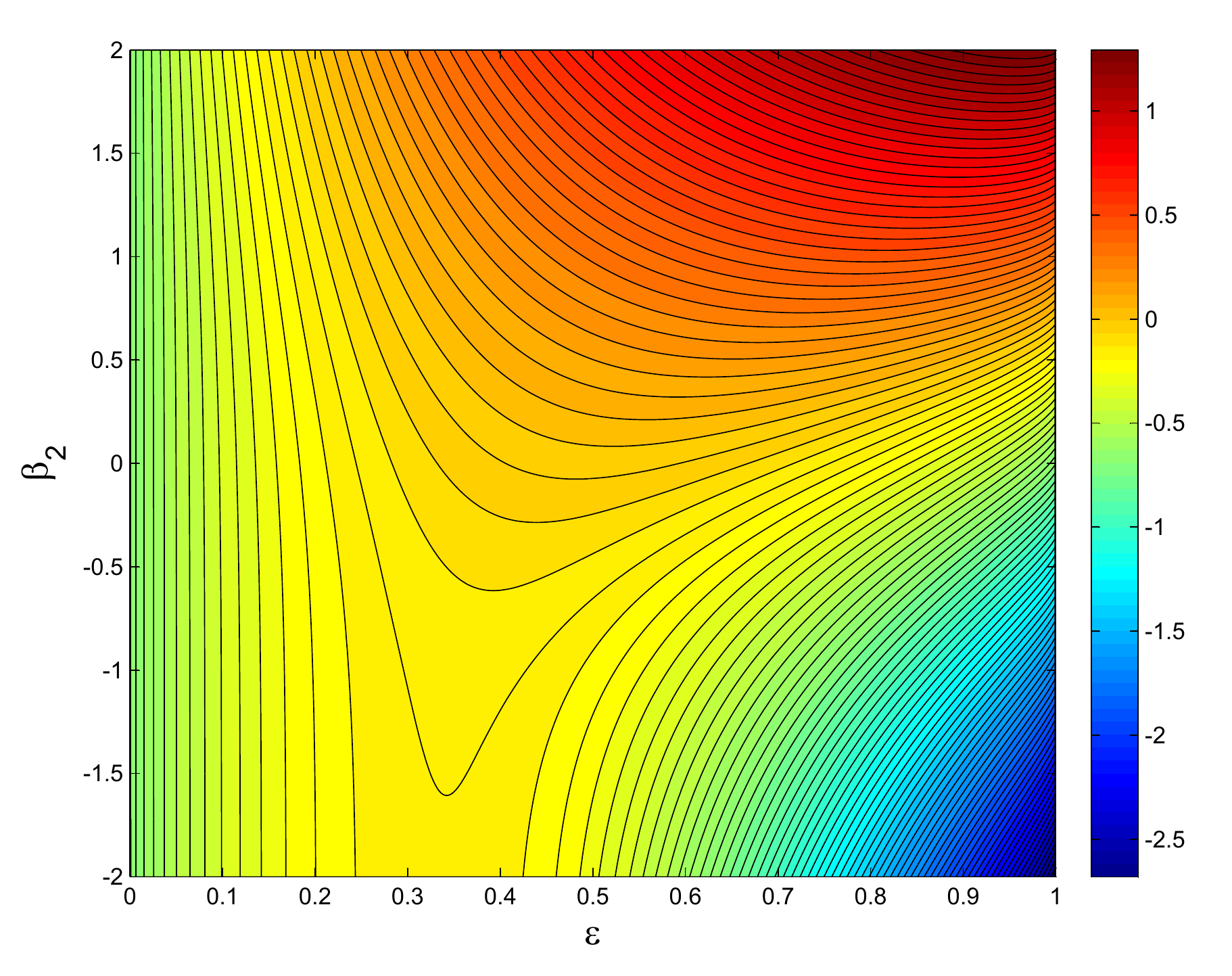}
\end{subfigure}
\begin{subfigure}
  \centering
  \includegraphics[height=1.785in, width=2.38in]{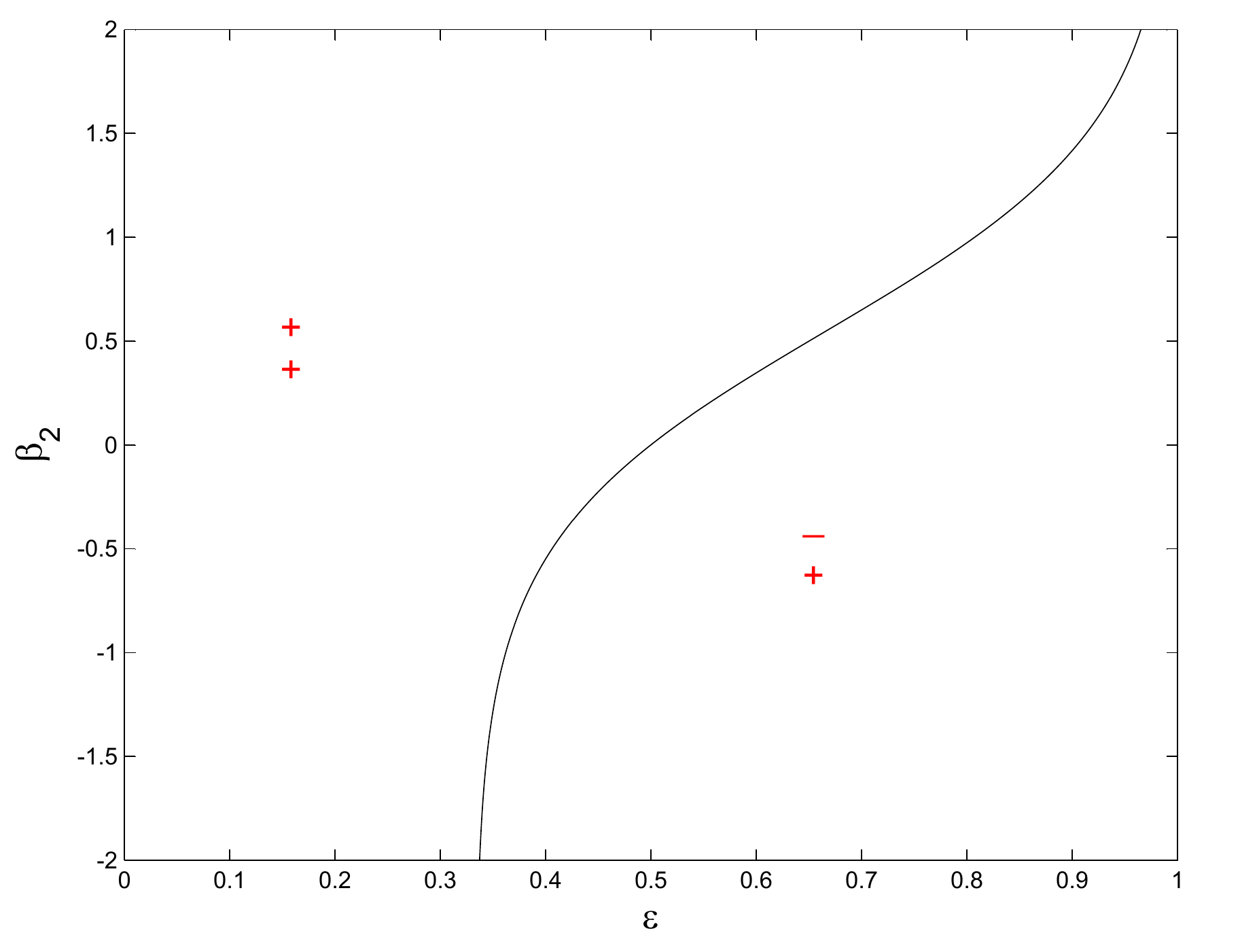}
\end{subfigure}
\caption{On the left hand side, we have the contour plot of the limiting conditional free energy density $\phi(\epsilon,\beta_{2})$ obtained
from Theorem \ref{exact2}. On the right hand side, we specify the regions of monotonicity as obtained in Remark \ref{mono23}.
$\phi$ is always increasing in $\beta_{2}$. In region $^{+}_{+}$, $\phi$ is increasing in $\epsilon$; in region $^{-}_{+}$, $\phi$ is decreasing in $\epsilon$. The boundary is specified in Remark \ref{mono23}.}
\end{figure}

\begin{remark}
\label{mono23}
Let us analyze the monotonicity of the two limiting conditional free energy densities. We have
\begin{align}
\frac{\partial\phi(\epsilon,\beta_{2})}{\partial\beta_{2}}=r^{\ast}+\left[\beta_2+\frac{\partial \lambda(\epsilon, r^*)}{\partial r^*}\right]\frac{\partial r^*}{\partial \beta_2}=r^*,
\\
\frac{\partial\psi(\beta_{1},r)}{\partial\beta_{1}}=\epsilon^{\ast}+\left[\beta_1+\frac{\partial \lambda(\epsilon^*,r)}{\partial \epsilon^*}\right]\frac{\partial \epsilon^*}{\partial \beta_1}=\epsilon^*.
\nonumber
\end{align}
Therefore $\phi(\epsilon,\beta_{2})$ and $\psi(\beta_{1},r)$ are increasing in $\beta_{2}$ and $\beta_{1}$ respectively. Moreover, we have
\begin{align}
\frac{\partial\phi(\epsilon,\beta_{2})}{\partial\epsilon}
&=\left[\beta_{2}+\frac{\partial\lambda(\epsilon,r^{\ast})}{\partial r^{\ast}}\right]\frac{\partial r^{\ast}}{\partial\epsilon}
+\frac{\partial\lambda(\epsilon,r^{\ast})}{\partial\epsilon}
\\
&=-\log\left(\frac{\epsilon-r^{\ast}}{1+r^{\ast}-2\epsilon}\right).
\nonumber
\end{align}
Therefore $\phi(\epsilon,\beta_{2})$ is increasing in $\epsilon$ if
and only if $1+2r^{\ast}\geq 3\epsilon$. This is equivalent to $\epsilon\leq \frac{1}{2}$ when $\beta_2=0$; while for $\beta_2\neq 0$, this is equivalent to
\begin{align}
\frac{(2\epsilon e^{2\beta_{2}}-2\epsilon+1)-\sqrt{(2\epsilon e^{2\beta_{2}}-2\epsilon+1)^{2}-4\epsilon^{2}e^{2\beta_{2}}(e^{2\beta_{2}}-1)}}
{e^{2\beta_{2}}-1}\geq 3\epsilon-1,
\end{align}
which can be simplified to
\begin{align}
&-\epsilon e^{2\beta_{2}}+\epsilon+e^{2\beta_{2}}\geq\sqrt{(2\epsilon e^{2\beta_{2}}-2\epsilon+1)^{2}-4\epsilon^{2}e^{2\beta_{2}}(e^{2\beta_{2}}-1)} \qquad\text{if $\beta_2>0$},
\\
&-\epsilon e^{2\beta_{2}}+\epsilon+e^{2\beta_{2}}\leq\sqrt{(2\epsilon e^{2\beta_{2}}-2\epsilon+1)^{2}-4\epsilon^{2}e^{2\beta_{2}}(e^{2\beta_{2}}-1)} \qquad\text{if $\beta_2<0$},
\nonumber
\end{align}
and can be further simplified to
\begin{align}
&(1-\epsilon)e^{4\beta_{2}}-2\epsilon e^{2\beta_{2}}+(3\epsilon-1)\geq 0 \qquad\text{if $\beta_2>0$},
\\
&(1-\epsilon)e^{4\beta_{2}}-2\epsilon e^{2\beta_{2}}+(3\epsilon-1)\leq 0 \qquad\text{if $\beta_2<0$},
\nonumber
\end{align}
or alternatively
\begin{align}
\epsilon\leq\frac{e^{2\beta_{2}}+1}{e^{2\beta_{2}}+3}.
\end{align}
Similarly,
\begin{align}
\frac{\partial\psi(\beta_{1},r)}{\partial r}
&=\left[\beta_{1}+\frac{\partial\lambda(\epsilon^{\ast},r)}{\partial\epsilon^{\ast}}\right]\frac{\partial\epsilon^{\ast}}{\partial r}
+\frac{\partial\lambda(\epsilon^{\ast},r)}{\partial r}
\\
&=-\frac{1}{2}\log\left(\frac{r(1+r-2\epsilon^{\ast})}{(\epsilon^{\ast}-r)^{2}}\right).
\nonumber
\end{align}
Therefore $\psi(\beta_{1},r)$ is increasing in $r$ if
and only if $r\leq(\epsilon^{\ast})^{2}$. This is equivalent to
\begin{equation}
\sqrt{r}\leq\frac{e^{\beta_{1}}(1+r)+r}{2e^{\beta_{1}}+1},
\end{equation}
or alternatively
\begin{equation}
\beta_{1}\geq\log\left(\frac{\sqrt{r}}{1-\sqrt{r}}\right).
\end{equation}
\end{remark}

\begin{figure}
\centering
\begin{subfigure}
  \centering
  \includegraphics[height=1.785in, width=2.38in]{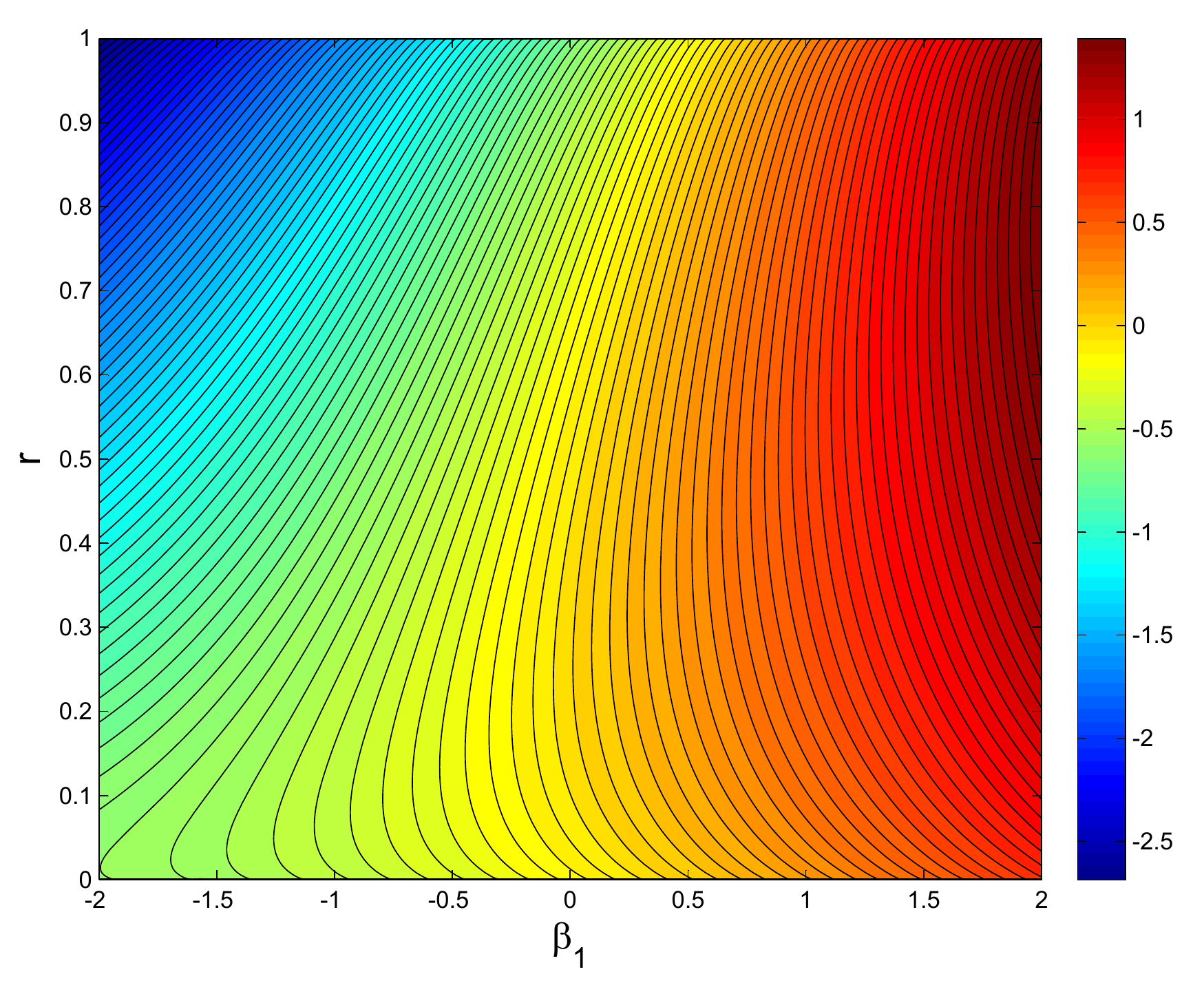}
\end{subfigure}
\begin{subfigure}
  \centering
  \includegraphics[height=1.785in, width=2.38in]{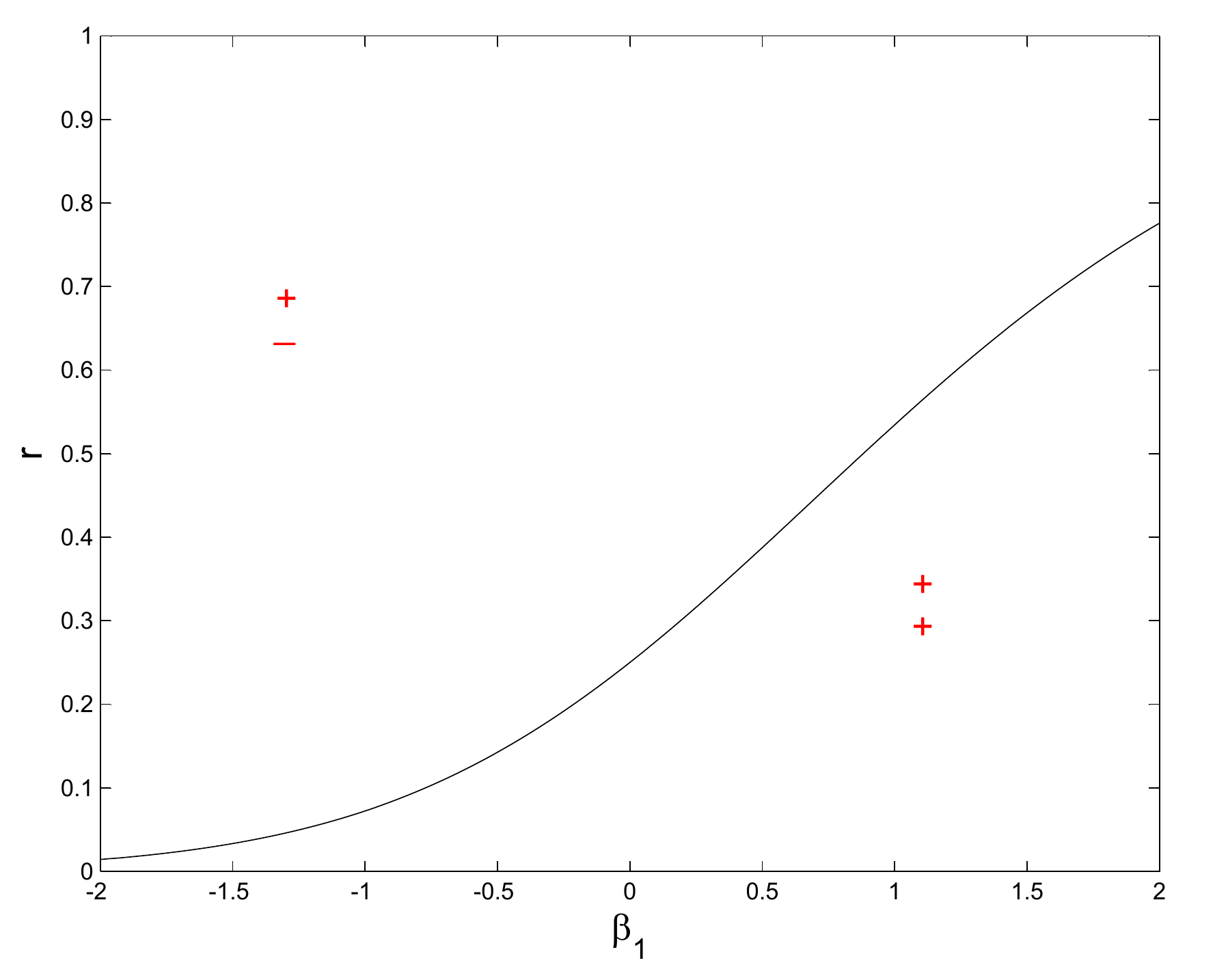}
\end{subfigure}
\caption{On the left hand side, we have the contour plot of the limiting conditional free energy density $\psi(\beta_{1},r)$ obtained
from Theorem \ref{exact2}. On the right hand side, we specify the regions of monotonicity as obtained in Remark \ref{mono23}.
$\psi$ is always increasing in $\beta_{1}$. In region $^{+}_{-}$, $\psi$ is decreasing in $r$; in region $^{+}_{+}$, $\psi$ is increasing in $r$. The boundary is specified in Remark \ref{mono23}.}
\end{figure}

\section{Another look at the grand canonical ensemble}
\label{grand}
A crucial observation on the reciprocal model is that the probability measure (\ref{1}) may be alternatively written as
\begin{equation}
\label{alt}
\mathbb{P}_n^{\beta_1,\beta_2}(X)=Z_n(\beta_1,\beta_2)^{-1}\prod_{1\leq i<j\leq n}e^{\beta_1(X_{ij}+X_{ji})+2\beta_2 X_{ij}X_{ji}},
\end{equation}
and is equivalent to an Erd\H{o}s-R\'{e}nyi type measure which assigns the following joint distribution iid for every pair $(i,j)$ with $1\leq i<j\leq n$:
\begin{equation}
\label{observe}
(X_{ij}, X_{ji})=\left\{
                                \begin{array}{ll}
                                  (0,0) & \hbox{with probability $\frac{1}{1+2e^{\beta_1}+e^{2\beta_1+2\beta_2}}$,} \\
                                  (0,1) & \hbox{with probability $\frac{e^{\beta_1}}{1+2e^{\beta_1}+e^{2\beta_1+2\beta_2}}$,} \\
                                  (1,0) & \hbox{with probability $\frac{e^{\beta_1}}{1+2e^{\beta_1}+e^{2\beta_1+2\beta_2}}$,} \\
                                  (1,1) & \hbox{with probability $\frac{e^{2\beta_1+2\beta_2}}{1+2e^{\beta_1}+e^{2\beta_1+2\beta_2}}$.}
                                \end{array}
                              \right.
\end{equation}
This tractable feature of the model has been partially used in earlier sections of the paper where we study the microcanonical and canonical ensembles. In this section we will explore further the consequences of the iid structure on the grand canonical ensemble. As seen in Corollary \ref{cor}, the entropy $\lambda(\epsilon,r)$ and the free energy $\chi(\beta_1,\beta_2)$ are related by the Legendre transform. An explicit connection between $(\epsilon,r)$ and $(\beta_1,\beta_2)$ is given in Theorem \ref{exact} (see (\ref{EqnI}) and (\ref{EqnII})). The next proposition, which easily follows from (\ref{alt}) and (\ref{observe}), calculates the mean of the edge and reciprocal densities when the parameters are fixed.

\begin{proposition}
\label{mean}
For any $i\neq 1$,
\begin{equation}
\label{binom1}
\mathbb{P}_n^{\beta_1,\beta_2}(X_{1i}=1)=\frac{e^{\beta_{1}}+e^{2\beta_{1}+2\beta_{2}}}{1+2e^{\beta_{1}}+e^{2\beta_{1}+2\beta_{2}}},
\end{equation}
\begin{equation}
\label{binom2}
\mathbb{P}_n^{\beta_1,\beta_2}(X_{1i}=1, X_{i1}=1)=\frac{e^{2\beta_{1}+2\beta_{2}}}{1+2e^{\beta_{1}}+e^{2\beta_{1}+2\beta_{2}}}.
\end{equation}
\end{proposition}

\begin{remark}
In the reciprocal model, the number of directed edges from a given vertex is Binomial with parameter given by (\ref{binom1}) and the number of reciprocal edges from a given vertex is Binomial with parameter given by (\ref{binom2}). This leads to a host of results in large deviations theory. For example,
\begin{equation}
\mathbb{P}_n^{\beta_1,\beta_2}\left(\sum_{i=2}^n X_{1i}>(n-1)\epsilon\right)\asymp \exp(-(n-1)I^{\beta_1,\beta_2}(\epsilon)),
\end{equation}
where the rate function
\begin{align}
\label{crate}
&I^{\beta_1,\beta_2}(\epsilon):=\sup_{\theta>0}\left\{\theta \epsilon-\log \mathbb{E}_n^{\beta_1,\beta_2}\left[e^{\theta X_{1i}}\right]\right\}\\
&=\epsilon\log \frac{\epsilon}{e^{\beta_1}+e^{2\beta_1+2\beta_2}}+(1-\epsilon)\log \frac{1-\epsilon}{1+e^{\beta_1}}+\log \left(1+2e^{\beta_1}+e^{2\beta_1+2\beta_2}\right).
\nonumber
\end{align}
Note that when $\beta_1=\beta_2=0$, (\ref{crate}) reduces to the rate function under the uniform measure,
\begin{equation}
I(\epsilon)=\epsilon\log \epsilon+(1-\epsilon)\log(1-\epsilon)+\log 2,
\end{equation}
coinciding with (\ref{rate}).
\end{remark}

We can further study the fluctuations of the edge and reciprocal densities around their mean, i.e., the central limit theorem.

\begin{proposition}\label{CLT}
Under the grand canonical measure (\ref{1}),
\begin{equation}
n\left(e(X)-\frac{e^{\beta_{1}}+e^{2\beta_{1}+2\beta_{2}}}{1+2e^{\beta_{1}}+e^{2\beta_{1}+2\beta_{2}}},
r(X)-\frac{e^{2\beta_{1}+2\beta_{2}}}{1+2e^{\beta_{1}}+e^{2\beta_{1}+2\beta_{2}}}\right)\rightarrow N\left(\mu,\Sigma\right)
\end{equation}
in distribution as $n\rightarrow\infty$, where
\begin{equation}\label{Sigma}
\mu:=
\left(
\begin{array}{c}
-\frac{e^{\beta_{1}}+e^{2\beta_{1}+2\beta_{2}}}{1+2e^{\beta_{1}}+e^{2\beta_{1}+2\beta_{2}}}
\\
-\frac{e^{2\beta_{1}+2\beta_{2}}}{1+2e^{\beta_{1}}+e^{2\beta_{1}+2\beta_{2}}}
\end{array}
\right),
\hspace{0.5cm}
\Sigma:=
\left(
\begin{array}{cc}
\Sigma_{11} & \Sigma_{12}
\\
\Sigma_{21} & \Sigma_{22}
\end{array}
\right),
\end{equation}
where
\begin{align}
&\Sigma_{11}:=\frac{e^{\beta_{1}}+2e^{2\beta_{1}+2\beta_{2}}}
{1+2e^{\beta_{1}}+e^{2\beta_{1}+2\beta_{2}}}
-2\left(\frac{e^{\beta_{1}}+e^{2\beta_{1}+2\beta_{2}}}{1+2e^{\beta_{1}}+e^{2\beta_{1}+2\beta_{2}}}\right)^{2},
\\
&\Sigma_{12}=\Sigma_{21}:=\frac{2e^{2\beta_{1}+2\beta_{2}}}
{1+2e^{\beta_{1}}+e^{2\beta_{1}+2\beta_{2}}}
-\frac{2e^{2\beta_{1}+2\beta_{2}}(e^{\beta_{1}}+e^{2\beta_{1}+2\beta_{2}})}{\left(1+2e^{\beta_{1}}+e^{2\beta_{1}+2\beta_{2}}\right)^{2}},
\nonumber
\\
&\Sigma_{22}:=\frac{2e^{2\beta_{1}+2\beta_{2}}}
{1+2e^{\beta_{1}}+e^{2\beta_{1}+2\beta_{2}}}
-2\left(\frac{e^{2\beta_{1}+2\beta_{2}}}{1+2e^{\beta_{1}}+e^{2\beta_{1}+2\beta_{2}}}\right)^{2}.
\nonumber
\end{align}
\end{proposition}

\begin{remark}
Note that when $\beta_{1}=\beta_{2}=0$, Proposition \ref{CLT} gives the central limit theorem
for $e(X)$ and $r(X)$ under the uniform measure and we have
\begin{equation}
\mu=
\left(
\begin{array}{c}
-\frac{1}{2}
\\
-\frac{1}{4}
\end{array}
\right),
\hspace{0.5cm}
\Sigma=
\left(
\begin{array}{cc}
\frac{1}{4} & \frac{1}{4}
\\
\frac{1}{4} & \frac{3}{8}
\end{array}
\right).
\end{equation}
\end{remark}

\begin{remark}
The drift term $\mu$ in Proposition \ref{CLT} is due to the definition of $e(X)$ and $r(X)$ in \eqref{es}.
If one defines $e(X)$ and $r(X)$ as
\begin{equation}\label{esII}
e(X)=\frac{1}{n(n-1)}\sum_{1\leq i,j\leq n}X_{ij},\quad
r(X)=\frac{1}{n(n-1)}\sum_{1\leq i,j\leq n}X_{ij}X_{ji}
\end{equation}
instead, then Proposition \ref{CLT} will hold with minor modifications:
\begin{equation}
(n-1)\left(e(X)-\frac{e^{\beta_{1}}+e^{2\beta_{1}+2\beta_{2}}}{1+2e^{\beta_{1}}+e^{2\beta_{1}+2\beta_{2}}},
r(X)-\frac{e^{2\beta_{1}+2\beta_{2}}}{1+2e^{\beta_{1}}+e^{2\beta_{1}+2\beta_{2}}}\right)\rightarrow N\left(0,\Sigma\right).
\end{equation}
Though definitions (\ref{es}) and \eqref{esII} lead to a difference of the drift term in the central limit theorem, they are indistinguishable
as regards the limiting entropy and free energy densities.
\end{remark}

\begin{proof}[Proof of Proposition \ref{CLT}]
For any $\theta_{1},\theta_{2}\in\mathbb{R}$,
\begin{align}
&\mathbb{E}_{n}^{\beta_{1},\beta_{2}}\left[e^{\theta_{1}n(e(X)-\frac{e^{\beta_{1}}+e^{2\beta_{1}+2\beta_{2}}}{1+2e^{\beta_{1}}+e^{2\beta_{1}+2\beta_{2}}})+\theta_{2}n(r(X)-\frac{e^{2\beta_{1}+2\beta_{2}}}{1+2e^{\beta_{1}}+e^{2\beta_{1}+2\beta_{2}}})}\right]
\\
&=
\frac{\mathbb{E}_n\left[e^{n^2\left((\frac{\theta_{1}}{n}+\beta_{1})e(X)+(\frac{\theta_{2}}{n}+\beta_{2})r(X)\right)}\right]
e^{-\frac{e^{\beta_{1}}+e^{2\beta_{1}+2\beta_{2}}}{1+2e^{\beta_{1}}+e^{2\beta_{1}+2\beta_{2}}}
\theta_{1}n
-\frac{e^{2\beta_{1}+2\beta_{2}}}{1+2e^{\beta_{1}}+e^{2\beta_{1}+2\beta_{2}}}\theta_{2}n}}
{\mathbb{E}_n\left[e^{n^2\left(\beta_1e(X)+\beta_2r(X)\right)}\right]}
\nonumber
\\
&=
\frac{\left(1+2e^{\frac{\theta_{1}}{n}+\beta_{1}}+e^{\frac{2\theta_{1}}{n}+\frac{2\theta_{2}}{n}+2\beta_{1}+2\beta_{2}}\right)^{\frac{n(n-1)}{2}}
e^{-\frac{e^{\beta_{1}}+e^{2\beta_{1}+2\beta_{2}}}{1+2e^{\beta_{1}}+e^{2\beta_{1}+2\beta_{2}}}
\theta_{1}n
-\frac{e^{2\beta_{1}+2\beta_{2}}}{1+2e^{\beta_{1}}+e^{2\beta_{1}+2\beta_{2}}}\theta_{2}n}}
{\left(1+2e^{\beta_{1}}+e^{2\beta_{1}+2\beta_{2}}\right)^{\frac{n(n-1)}{2}}}
\nonumber
\\
&=\left(1+\frac{a}{n}+\frac{b}{n^{2}}+O(n^{-3})\right)^{\frac{n(n-1)}{2}}
e^{-\frac{e^{\beta_{1}}+e^{2\beta_{1}+2\beta_{2}}}{1+2e^{\beta_{1}}+e^{2\beta_{1}+2\beta_{2}}}
\theta_{1}n
-\frac{e^{2\beta_{1}+2\beta_{2}}}{1+2e^{\beta_{1}}+e^{2\beta_{1}+2\beta_{2}}}\theta_{2}n},
\nonumber
\end{align}
where
\begin{align}
&a:=
\frac{2(e^{\beta_{1}}+e^{2\beta_{1}+2\beta_{2}})\theta_1
+2e^{2\beta_{1}+2\beta_{2}}\theta_2}
{1+2e^{\beta_{1}}+e^{2\beta_{1}+2\beta_{2}}},
\\
&b:=
\frac{(e^{\beta_{1}}+2e^{2\beta_{1}+2\beta_{2}})\theta_1^2
+4e^{2\beta_{1}+2\beta_{2}}\theta_{1}\theta_{2}
+2e^{2\beta_{1}+2\beta_{2}}\theta_2^2}
{1+2e^{\beta_{1}}+e^{2\beta_{1}+2\beta_{2}}}.
\end{align}
Since
\begin{equation}
\log(1+\frac{a}{n}+\frac{b}{n^{2}}+O(n^{-3}))=\frac{a}{n}+\frac{b-a^{2}/2}{n^{2}}+O(n^{-3}),
\end{equation}
we have
\begin{equation}
\frac{n(n-1)}{2}\log\left(1+\frac{a}{n}+\frac{b}{n^{2}}+O(n^{-3})\right)=\frac{a}{2}n-\frac{a}{2}+\frac{b-a^2/2}{2}+O(n^{-1}).
\end{equation}
This implies that
\begin{align}
&\mathbb{E}_{n}^{\beta_{1},\beta_{2}}\left[e^{\theta_{1}n(e(X)-\frac{e^{\beta_{1}}+e^{2\beta_{1}+2\beta_{2}}}{1+2e^{\beta_{1}}+e^{2\beta_{1}+2\beta_{2}}})+\theta_{2}n(r(X)-\frac{e^{2\beta_{1}+2\beta_{2}}}{1+2e^{\beta_{1}}+e^{2\beta_{1}+2\beta_{2}}})}\right]
\\
&
\rightarrow\exp\bigg\{-\frac{e^{\beta_{1}}+e^{2\beta_{1}+2\beta_{2}}}{1+2e^{\beta_{1}}+e^{2\beta_{1}+2\beta_{2}}}
\theta_{1}
-\frac{e^{2\beta_{1}+2\beta_{2}}}{1+2e^{\beta_{1}}+e^{2\beta_{1}+2\beta_{2}}}\theta_{2}
\nonumber
\\
&\qquad\qquad
+\frac{1}{2}\left(\frac{e^{\beta_{1}}+2e^{2\beta_{1}+2\beta_{2}}}
{1+2e^{\beta_{1}}+e^{2\beta_{1}+2\beta_{2}}}
-2\left(\frac{e^{\beta_{1}}+e^{2\beta_{1}+2\beta_{2}}}{1+2e^{\beta_{1}}+e^{2\beta_{1}+2\beta_{2}}}\right)^{2}\right)\theta_{1}^{2}
\nonumber
\\
&\qquad\qquad
+\left(\frac{2e^{2\beta_{1}+2\beta_{2}}}
{1+2e^{\beta_{1}}+e^{2\beta_{1}+2\beta_{2}}}
-\frac{2e^{2\beta_{1}+2\beta_{2}}(e^{\beta_{1}}+e^{2\beta_{1}+2\beta_{2}})}{\left(1+2e^{\beta_{1}}+e^{2\beta_{1}+2\beta_{2}}\right)^{2}}\right)\theta_{1}\theta_{2}
\nonumber
\\
&\qquad\qquad
+\frac{1}{2}\left(\frac{2e^{2\beta_{1}+2\beta_{2}}}
{1+2e^{\beta_{1}}+e^{2\beta_{1}+2\beta_{2}}}
-2\left(\frac{e^{2\beta_{1}+2\beta_{2}}}{1+2e^{\beta_{1}}+e^{2\beta_{1}+2\beta_{2}}}\right)^{2}\right)\theta_{2}^{2}
\bigg\}
\nonumber
\end{align}
as $n\rightarrow\infty$. Since convergence of moment generating functions implies convergence in distribution, the proof is complete.
\end{proof}

Similar to the analysis in Yin and Zhu \cite{YinZhu}, we can also study directed graphs where the parameters depend on the number of vertices. Assume that $\beta_{1}^{(n)}=\alpha_{n}\beta_{1}$ and $\beta_{2}^{(n)}=\alpha_{n}\beta_{2}$, where $\beta_1$ and $\beta_2$ are fixed, $\alpha_{n}$ is positive and $\alpha_{n}\rightarrow\infty$ as $n\rightarrow\infty$. With some abuse of notation, we will still denote the associated normalization constant and probability measure by $Z_n(\beta_1,\beta_2)$ and $\mathbb{P}_n^{\beta_1,\beta_2}$ respectively.
From the proof of Theorem \ref{Lyapunov},
\begin{equation}\label{Zn}
Z_n(\beta_1,\beta_2)^{1/\binom{n}{2}}=1+2e^{\alpha_n\beta_{1}}+e^{2\alpha_n(\beta_{1}+\beta_{2})},
\end{equation}
which yields the following asymptotics for the normalization.

\begin{proposition}
\label{scaling}
(i) When $\beta_{1}<0$ and $\beta_{1}+\beta_{2}<0$, $\lim_{n\rightarrow\infty}(Z_n(\beta_1,\beta_2))^{\frac{1}{n^{2}}}=1$.

(ii) When $\beta_{1}<0$ and $\beta_{1}+\beta_{2}=0$, $\lim_{n\rightarrow\infty}(Z_n(\beta_1,\beta_2))^{\frac{1}{n^{2}}}=\sqrt{2}$.

(iii) When $\beta_{1}\leq 0$ and $\beta_{1}+\beta_{2}>0$, $\lim_{n\rightarrow\infty}\frac{(Z_n(\beta_1,\beta_2))^{\frac{1}{n^{2}}}}{e^{\alpha_{n}(\beta_{1}+\beta_{2})}}=1$.

(iv) When $\beta_{1}=0$ and $\beta_{2}<0$, $\lim_{n\rightarrow\infty}(Z_n(\beta_1,\beta_2))^{\frac{1}{n^{2}}}=\sqrt{3}$.

(v) When $\beta_{1}=0$ and $\beta_{2}=0$, $\lim_{n\rightarrow\infty}(Z_n(\beta_1,\beta_2))^{\frac{1}{n^{2}}}=2$.

(vi) When $\beta_{1}>0$ and $\beta_{1}+2\beta_{2}<0$,
$\lim_{n\rightarrow\infty}\frac{(Z_n(\beta_1,\beta_2))^{\frac{1}{n^{2}}}}{e^{\frac{1}{2}\alpha_{n}\beta_{1}}}=\sqrt{2}$.

(vii) When $\beta_{1}>0$ and $\beta_{1}+2\beta_{2}=0$,
$\lim_{n\rightarrow\infty}\frac{(Z_n(\beta_1,\beta_2))^{\frac{1}{n^{2}}}}{e^{\frac{1}{2}\alpha_{n}\beta_{1}}}=\sqrt{3}$.

(viii) When $\beta_{1}>0$ and $\beta_{1}+2\beta_{2}>0$,
$\lim_{n\rightarrow\infty}\frac{(Z_n(\beta_1,\beta_2))^{\frac{1}{n^{2}}}}{e^{\alpha_{n}(\beta_{1}+\beta_{2})}}=1$.
\end{proposition}

Since many networks data are sparse in the real world, we are more interested in the situation where a random graph sampled from this modified model is sparse, i.e., the probability that there is an edge between vertex $i$ and vertex $j$ goes to $0$ as $n\rightarrow \infty$. One natural question to ask is for what set of parameters $(\beta_1, \beta_2)$ will this happen? And a natural follow-up question is what is the speed of the graph towards sparsity when this indeed happens? We give some concrete answers to these questions.

\begin{proposition}
\label{sparsemean}
For any $i\neq 1$,

(i) When $\beta_{1}<2(\beta_{1}+\beta_{2})<0$, $\lim_{n\rightarrow\infty}\frac{\mathbb{P}_n^{\beta_1,\beta_2}(X_{1i}=1)}{
e^{2\alpha_{n}(\beta_{1}+\beta_{2})}}=\frac{1}{4}$.

(ii) When $2(\beta_{1}+\beta_{2})<\beta_{1}<0$, $\lim_{n\rightarrow\infty}\frac{\mathbb{P}_n^{\beta_1,\beta_2}(X_{1i}=1)}{
e^{\alpha_{n}\beta_{1}}}=\frac{1}{4}$.

(iii) When $\beta_{1}=2(\beta_{1}+\beta_{2})<0$, $\lim_{n\rightarrow\infty}\frac{\mathbb{P}_n^{\beta_1,\beta_2}(X_{1i}=1)}{
e^{\alpha_{n}\beta_{1}}}=\frac{1}{2}$.
\end{proposition}

\begin{proof}[Proof of Proposition \ref{sparsemean}]
From Proposition \ref{mean}, $\mathbb{P}_n^{\beta_1,\beta_2}(X_{1i}=1)\rightarrow 0$ as $n\rightarrow \infty$ only when $\beta_{1}<0$ and $\beta_{1}+\beta_{2}<0$. The rest of the proof easily follows.
\end{proof}

\section{Further discussion}
\label{discussion}
We have studied directed graphs whose sufficient statistics are given by edge and reciprocal densities. Now let us generalize these ideas and analyze directed graphs whose sufficient statistics also include densities of \emph{reciprocal $p$-stars} and \emph{reciprocal triangles}. Reciprocal triangles are sometimes called \emph{cyclic triads} in the literature, see e.g. Robins et al. \cite{Robins}. They are used
to model the situation where you have three vertices $i$, $j$ and $k$ and there are bilateral relations between $i$ and $j$,
$j$ and $k$, and $k$ and $i$, i.e., $X_{ij}=X_{ji}=X_{jk}=X_{kj}=X_{ki}=X_{ik}=1$. Similarly, reciprocal $p$-stars have generated significant interest as well.
We define the densities of reciprocal triangles and reciprocal $p$-stars respectively as
\begin{equation}
t(X):=\frac{1}{n^{3}}\sum_{1\leq i,j,k\leq n}X_{ij}X_{ji}X_{jk}X_{kj}X_{ki}X_{ik}
\end{equation}
and
\begin{equation}
s(X):=\frac{1}{n^{p+1}}\sum_{i=1}^{n}\left(\sum_{j=1}^{n}X_{ij}X_{ji}\right)^{p}.
\end{equation}

As for the less complicated reciprocal model investigated earlier, we are interested in the \textit{limiting free energy density}
\begin{align}\label{MacroLimit}
&\chi(\beta_{1},\beta_{2},\beta_{3},\beta_{4})
:=\lim_{n\rightarrow\infty}\frac{1}{n^{2}}\log Z_n(\beta_1,\beta_2,\beta_3,\beta_4)\\
&=\lim_{n\rightarrow\infty}\frac{1}{n^{2}}\log2^{n(n-1)}\mathbb{E}_n\left[e^{n^2\left(\beta_{1}e(X)+\beta_{2}r(X)+\beta_{3}t(X)+\beta_{4}s(X)\right)}\right]
\nonumber
\end{align}
for the grand canonical ensemble and the \textit{limiting entropy density}
\begin{align}\label{MicroLimit}
&\lambda(\epsilon,r,t,s):=\lim_{\delta\rightarrow 0}\lim_{n\rightarrow\infty}\lambda_{n,\delta}(\epsilon,r,t,s)\\
&=\lim_{\delta\rightarrow 0}\lim_{n\rightarrow\infty}\frac{1}{n^{2}}\log\mathbb{P}_n\left(e(X)\in B_{\delta}(\epsilon),r(X)\in B_{\delta}(r),t(X)\in B_{\delta}(t),s(X)\in B_{\delta}(s)\right)
\nonumber
\end{align}
for the microcanonical ensemble, where $B_{\delta}(x):=\{y:|y-x|<\delta\}$. In contrast to the reciprocal model discussed in Sections \ref{microcanonical}-\ref{grand}, in this generalized model, different pairs of vertices are no longer independent, and this sophistication renders a concrete analysis of the model rather difficult. We hope the partial answers presented in this section will provide insight into its intrinsic structure and help us better understand the nature of reciprocity. Previously, we derive the limiting free energy density (\ref{chi}) and then obtain the limiting entropy density (\ref{sup}) using the Legendre transform. We take an opposite approach here. Below we find an expression for the limiting entropy density (\ref{MicroLimit}) and then apply the inverse Legendre transform to develop an expression for the limiting free energy density (\ref{MacroLimit}).

We examine the limiting entropy density \eqref{MicroLimit} first. A key observation is that we can define $Z_{ij}=Z_{ji}=X_{ij}X_{ji}$ so that
$(Z_{ij})_{1\leq i<j \leq n}$ are iid random variables with $\mathbb{P}_n(Z_{ij}=1)=\frac{1}{4}$
and $\mathbb{P}_n(Z_{ij}=0)=\frac{3}{4}$. Then densities of reciprocal edges, reciprocal triangles, and reciprocal $p$-stars may be alternatively written as
\begin{align}
&r(X)=\frac{1}{n^{2}}\sum_{1\leq i,j\leq n}Z_{ij},
\\
&t(X)=\frac{1}{n^{3}}\sum_{1\leq i,j,k\leq n}Z_{ij}Z_{jk}Z_{ki},
\nonumber
\\
&s(X)=\frac{1}{n^{p+1}}\sum_{i=1}^{n}\left(\sum_{j=1}^{n}Z_{ij}\right)^{p}.
\nonumber
\end{align}
Using Chatterjee and Varadhan's large deviations results for the Erd\H{o}s-R\'{e}nyi random graph, see e.g. Chatterjee and Varadhan \cite{ChatterjeeVaradhan}, this gives
\begin{align}\label{LDPI}
&\lim_{\delta\rightarrow 0}\lim_{n\rightarrow\infty}\frac{1}{n^2}\log\mathbb{P}_n\left(r(X)\in B_{\delta}(r),t(X)\in B_{\delta}(t),s(X)\in B_{\delta}(s)\right)
\\
&=-\inf_{\substack{g:[0,1]^{2}\rightarrow[0,1],g(x,y)=g(y,x)
\\
r(g)=r, t(g)=t,s(g)=s}}\frac{1}{2}I_{\frac{1}{4}}(g),
\nonumber
\end{align}
where
\begin{align}
&r(g):=\iint_{[0,1]^{2}}g(x,y)dxdy,
\\
&t(g):=\iiint_{[0,1]^{3}}g(x,y)g(y,z)g(z,x)dxdydz,
\nonumber
\\
&s(g):=\int_{0}^{1}\left(\int_{0}^{1}g(x,y)dy\right)^{p}dx,
\nonumber
\end{align}
and $I_{\frac{1}{4}}(g):=\iint_{[0,1]^{2}}I_{\frac{1}{4}}(g(x,y))dxdy$, where
\begin{align}
I_{\frac{1}{4}}(x):&=x\log\left(\frac{x}{1/4}\right)+(1-x)\log\left(\frac{1-x}{1-1/4}\right)
\\
&=x\log 3+x\log x+(1-x)\log(1-x)-\log(3/4).
\nonumber
\end{align}

Another key observation is that
the distribution of $e(X)$ conditional on $(Z_{ij})_{1\leq i,j\leq n}$ is the same
as conditional on $r(X)$. Thus, the distribution of $e(X)$ conditional on $r(X),t(X),s(X)$ is the same
as conditional on $r(X)$. We compute
\begin{align}
&\lim_{\delta\rightarrow 0}\lim_{n\rightarrow\infty}\frac{1}{n^{2}}\log\mathbb{P}_n(r(X)\in B_{\delta}(r))
\\
&=-\frac{1}{2}r\log\left(\frac{r}{1/4}\right)-\frac{1}{2}(1-r)\log\left(\frac{1-r}{1-1/4}\right),
\nonumber
\end{align}
which, combined with (\ref{epsilonr}), implies that
\begin{align}
\zeta(\epsilon,r)&:=\lim_{\delta\rightarrow 0}\lim_{n\rightarrow\infty}\frac{1}{n^{2}}\log\mathbb{P}_n(e(X)\in B_{\delta}(\epsilon)|r(X)\in B_{\delta}(r))
\\
&=-\epsilon\log\left(\frac{\epsilon-r}{1+r-2\epsilon}\right)
-\frac{r}{2}\log\left(\frac{(1-r)(1+r-2\epsilon)}{3(\epsilon-r)^{2}}\right)
\nonumber
\\
&\qquad\qquad\qquad
+\frac{1}{2}\log\left(\frac{1-r}{3(1+r-2\epsilon)}\right).
\nonumber
\end{align}
Together with \eqref{LDPI}, we hence conclude that
\begin{equation}
\lambda(\epsilon,r,t,s)=\zeta(\epsilon,r)-\inf_{\substack{g:[0,1]^{2}\rightarrow[0,1],g(x,y)=g(y,x)
\\
r(g)=r, t(g)=t,s(g)=s}}\frac{1}{2}I_{\frac{1}{4}}(g).
\end{equation}

Next we examine the limiting free energy density (\ref{MacroLimit}). Varadhan's lemma in large deviations theory, see e.g. Dembo and Zeitouni \cite{Dembo}, states that the free energy $\chi(\beta_{1},\beta_{2},\beta_{3},\beta_{4})$ may be obtained as the inverse Legendre transform of the entropy $\lambda(\epsilon,r,t,s)$:
\begin{align}
&\chi(\beta_{1},\beta_{2},\beta_{3},\beta_{4})
\\
&=\sup_{0\leq \epsilon,r,t,s\leq 1}\left\{\beta_{1}\epsilon+\beta_{2}r+\beta_{3}t+\beta_{4}s+\lambda(\epsilon,r,t,s)+\log 2\right\}
\nonumber
\\
&=\sup_{\substack{0\leq \epsilon,r,t,s\leq 1
\\
g:[0,1]^{2}\rightarrow[0,1],g(x,y)=g(y,x)
\\
r(g)=r, t(g)=t,s(g)=s}}\left\{\beta_{1}\epsilon+\beta_{2}r+\beta_{3}t+\beta_{4}s+\zeta(\epsilon,r)-\frac{1}{2}\left(I_{\frac{1}{4}}(g)-2\log 2\right)\right\}.
\nonumber
\end{align}
Consider the optimization problem $\eta(\beta_{1},r):=\sup_{0\leq\epsilon\leq 1}\{\beta_{1}\epsilon+\zeta(\epsilon,r)\}$. Note that at optimality,
\begin{equation}
\frac{\partial\zeta(\epsilon,r)}{\partial\epsilon}
=-\log\left(\frac{\epsilon-r}{1+r-2\epsilon}\right)=-\beta_{1},
\end{equation}
which implies that $\epsilon=\frac{(1+r)e^{\beta_{1}}+r}{1+2e^{\beta_{1}}}$.
Therefore, we have
\begin{equation}
\eta(\beta_{1},r)=\beta_{1}\frac{(1+r)e^{\beta_{1}}+r}{1+2e^{\beta_{1}}}+\zeta\left(\frac{(1+r)e^{\beta_{1}}+r}{1+2e^{\beta_{1}}},r\right),
\end{equation}
and hence
\begin{align}
&\chi(\beta_{1},\beta_{2},\beta_{3},\beta_{4})
\\
&=\sup_{\substack{0\leq r,t,s\leq 1
\\
g:[0,1]^{2}\rightarrow[0,1],g(x,y)=g(y,x)
\\
r(g)=r, t(g)=t,s(g)=s}}\left\{\eta(\beta_{1},r)+\beta_{2}r+\beta_{3}t+\beta_{4}s-\frac{1}{2}\left(I_{\frac{1}{4}}(g)-2\log 2\right)\right\}
\nonumber
\\
&=\sup_{\substack{g:[0,1]^{2}\rightarrow[0,1]
\\
g(x,y)=g(y,x)}}\left\{\eta(\beta_{1},r(g))+\beta_{2}r(g)+\beta_{3}t(g)+\beta_{4}s(g)-\frac{1}{2}\left(I_{\frac{1}{4}}(g)-2\log 2\right)\right\}.
\nonumber
\end{align}
This is a complicated variational problem that is hard to solve in general, however we can proceed further in two special situations. As we will see, introducing reciprocal triangles and reciprocal $p$-stars into the probability measure significantly affects the structure of the model. Unlike the limiting free energy density $\chi(\beta_1,\beta_2)$ for the reciprocal model that only takes into account reciprocal edges, the limiting free energy density $\chi(\beta_1,\beta_2,\beta_3,\beta_4)$ for this generalized model may lose analyticity, giving rise to phase transitions in the grand canonical ensemble.

The first special situation is when $\beta_{1}=0$,
\begin{align}
\label{FormulaTwo}
&\chi(0,\beta_{2},\beta_{3},\beta_{4})
\\
&=\sup_{\substack{g:[0,1]^{2}\rightarrow[0,1]
\\
g(x,y)=g(y,x)}}
\left\{\beta_{2}r(g)+\beta_{3}t(g)+\beta_{4}s(g)-\frac{1}{2}\left(I_{\frac{1}{4}}(g)-2\log 2\right)\right\}
\nonumber
\\
&=\sup_{\substack{g:[0,1]^{2}\rightarrow[0,1]
\\
g(x,y)=g(y,x)}}
\left\{\left(\beta_{2}-\frac{\log 3}{2}\right)r(g)+\beta_{3}t(g)+\beta_{4}s(g)-\frac{1}{2}\left(I_{0}(g)-\log 3\right)\right\},
\nonumber
\end{align}
where $I_{0}(g):=\iint_{[0,1]^{2}}I_{0}(g(x,y))dxdy$ and $I_{0}(x):=x\log x+(1-x)\log(1-x)$.
This shows that $\chi(0,\beta_{2},\beta_{3},\beta_{4})$ may be equivalently viewed as the limiting free energy density of an undirected model whose sufficient statistics are given by (undirected) edge, triangle, and $p$-star densities. The $3$ parameters $\beta_2,\beta_3,\beta_4$ allow one to adjust the influence of these different local features on the limiting probability distribution and thus expectedly should impact the global structure of a random graph drawn from the model. It is therefore important to understand if and when the supremum in (\ref{FormulaTwo}) is attained and whether it is unique. Many people have delved into this area. A particularly significant discovery was made by Chatterjee and Diaconis \cite{ChatterjeeDiaconis}, who showed that the supremum in (\ref{FormulaTwo}) is always attained and a random graph drawn from the model must lie close to the maximizing set with probability vanishing in $n$. When $\beta_3,\beta_4\geq 0$, Yin \cite{Yin} further showed that the $3$-parameter space would consist of a single phase with first-order phase transition(s) across one (or more) surfaces, where all the first derivatives of $\chi$ exhibit (jump) discontinuities, and second-order phase transition(s) along one (or more) critical curves, where all the second derivatives of $\chi$ diverge.

The second special situation is when $\beta_{3}=0$,
\begin{align}
&\chi(\beta_{1},\beta_{2},0,\beta_{4})
\\
&=\sup_{\substack{g:[0,1]^{2}\rightarrow[0,1]
\\
g(x,y)=g(y,x)}}\left\{\eta(\beta_{1},r(g))+\beta_{2}r(g)+\beta_{4}s(g)-\frac{1}{2}\left(I_{\frac{1}{4}}(g)-2\log 2\right)\right\}.
\nonumber
\end{align}
We can derive the Euler-Lagrange equation for this variational problem, and it is given by
\begin{equation}
2\frac{\partial\eta}{\partial r}(\beta_{1},r(g))+2\beta_{2}+\beta_{4}pd(x)^{p-1}+\beta_{4}pd(y)^{p-1}=\log\left(\frac{g(x,y)}{1-g(x,y)}\right)+\log 3,
\end{equation}
where $d(x):=\int_{0}^{1}g(x,y)dy$. Solving for $g(x,y)$ and then integrating over $y$, we get
\begin{equation}
d(x)=\int_{0}^{1}\frac{dy}{1+3e^{-2\frac{\partial\eta}{\partial r}(\beta_{1},r(g))-2\beta_{2}-\beta_{4}pd(x)^{p-1}-\beta_{4}pd(y)^{p-1}}}.
\end{equation}
Following similar arguments as in Kenyon et al. \cite{Kenyon},
we conclude that $d(x)$ can take only finitely many values. The optimal graphon $g$ is multipodal and phase transitions are expected.

\section*{Acknowledgements}

The authors are very grateful to the anonymous referees for their invaluable suggestions
that greatly improved the quality of this paper.
Mei Yin's research was partially supported by NSF grant DMS-1308333. She appreciated the
opportunity to talk about this work in the Special Session on Spectral Theory, Disorder, and Quantum Many Body Physics at the 2015 AMS Central Spring Sectional Meeting, organized by Peter D. Hislop and Jeffrey Schenker.


\begin{thebibliography}{99}

\bibitem{Albert}
Albert, R., Jeong, H. and A.-L. Barab\'{a}si. (1999).
Diameter of the world wide web.
\textit{Nature}
\textbf{401}, 130-131.

\bibitem{AristoffZhu}
Aristoff, D. and L. Zhu. (2014).
On the phase transition curve in a directed exponential random graph model.
\textit{arXiv:} 1404.6514.

\bibitem{AristoffZhuII}
Aristoff, D. and L. Zhu. (2014).
Asymptotic structure and singularities in constrained directed graphs.
\textit{arXiv:} 1405.2466.

\bibitem{Besag}
Besag, J. (1975).
Statistical analysis of non-lattice data.
\textit{J. R. Stat. Soc. Ser. D. Stat.}
\textbf{24}, 179-195.

\bibitem{ChatterjeeDembo}
Chatterjee, S. and A. Dembo. (2014).
Nonlinear large deviations.
\textit{arXiv:} 1401.3495.

\bibitem{ChatterjeeDiaconis}
Chatterjee, S. and P. Diaconis. (2013).
Estimating and understanding exponential random graph models.
\textit{Ann. Statist.}
\textbf{41}, 2428-2461.

\bibitem{ChatterjeeVaradhan}
Chatterjee, S. and S. R. S. Varadhan. (2011).
The large deviation principle for the Erd\H{o}s-R\'{e}nyi random graph.
\textit{European J. Combin.}
\textbf{32}, 1000-1017.

\bibitem{Cheng}
Cheng, J., Romero, D., Meeder, B. and J. Kleinberg. (2011).
Predicting reciprocity in social networks.
In: \textit{IEEE Third International Conference on Social Computing}, pp. 49-56.

\bibitem{Dembo} Dembo, A. and O. Zeitouni. (1998).
\textit{Large Deviations Techniques and Applications (Second Edition)}.
Springer, New York.

\bibitem{FienbergI}
Fienberg, S. E. (2010). Introduction to papers on the modeling and analysis of network data.
\textit{Ann. Appl. Statist.}
\textbf{4}, 1-4.

\bibitem{FienbergII}
Fienberg, S. E. (2010). Introduction to papers on the modeling and analysis of network data--II.
\textit{Ann. Appl. Statist.}
\textbf{4}, 533-534.

\bibitem{Garlaschelli}
Garlaschelli, D. and M. I. Loffredo. (2004).
Patterns of link reciprocity in directed networks.
\textit{Phys. Rev. Lett.}
\textbf{93}, 268701.

\bibitem{GarlaschelliII}
Garlaschelli, D. and M. I. Loffredo. (2005).
Structure and evolution of the world trade network.
\textit{Phys. A}
\textbf{355}, 138-144.

\bibitem{Gleditsch}
Gleditsch, K. S. (2002).
Expanded trade and GDP data.
\textit{J. Conflict Resolut.}
\textbf{46}, 712-724.

\bibitem{Holland}
Holland, P. W. and S. Leinhardt. (1981).
An exponential family of probability distributions for directed graphs (with discussion).
\textit{J. Amer. Statist. Assoc.}
\textbf{76}, 33-65.

\bibitem{Jeong}
Jeong, H., Tombor, B., Albert, R., Oltvai, Z. N. and A.-L. Barab\'{a}si. (2000).
The large-scale organization of metabolic networks.
\textit{Nature}
\textbf{407},  651-655.

\bibitem{Kenyon}
Kenyon, R., Radin, C., Ren K. and L. Sadun. (2014).
Multipodal structure and phase transitions in large constrained graphs.
\textit{arXiv:} 1405.0599.

\bibitem{KenyonYin}
Kenyon, R. and M. Yin. (2014).
On the asymptotics of constrained exponential random graphs.
\textit{arXiv:} 1406.3662.

\bibitem{Lazega}
Lazega, E. and M. A. J. van Duijn. (1997).
Position in formal structure, personal characteristics and choices of advisors in a law firm: a logistic regression model for dyadic network data.
\textit{Soc. Netw.}
\textbf{19}, 375-397.

\bibitem{Lovasz2009}
Lov\'{a}sz, L. (2009).
Very large graphs.
\textit{Curr. Dev. Math.}
\textbf{2008}, 67-128.

\bibitem{Lov}
Lov\'{a}sz, L. (2012).
\textit{Large Networks and Graph Limits}.
American Mathematical Society, Providence.

\bibitem{LubetzkyZhao}
Lubetzky, E. and Y. Zhao. (2012).
On replica symmetry of large deviations in random graphs. \textit{arXiv:} 1210.7013.

\bibitem{LubetzkyZhaoII}
Lubetzky, E. and Y. Zhao. (2014).
On the variational problem for upper tails in sparse random graphs. \textit{arXiv:} 1402.6011.

\bibitem{Newman}
Newman, M. E. J. (2010).
\textit{Networks: An Introduction}.
Oxford University Press, Oxford.

\bibitem{NewmanII}
Newman, M. E. J., Forrest, S. and J. Balthrop. (2002).
Email networks and the spread of computer viruses.
\textit{Phys. Rev. E}
\textbf{66}, 035101.

\bibitem{RadinIII}
Radin, C., Ren, K. and L. Sadun. (2014).
The asymptotics of large constrained graphs.
\textit{J. Phys. A: Math. Theor.}
\textbf{47}, 175001.

\bibitem{RadinII}
Radin, C. and L. Sadun. (2013).
Phase transitions in a complex network.
\textit{J. Phys. A: Math. Theor.}
\textbf{46}, 305002.

\bibitem{RadinIV}
Radin, C. and L. Sadun. (2015).
Singularities in the entropy of asymptotically large simple graphs.
\textit{J. Stat. Phys.}
\textbf{158}, 853-865.

\bibitem{Radin}
Radin, C. and M. Yin. (2013).
Phase transitions in exponential random graphs.
\textit{Ann. Appl. Probab.}
\textbf{23}, 2458-2471.

\bibitem{Rinaldo}
Rinaldo, A., Fienberg, S. and Y. Zhou. (2009).
On the geometry of discrete exponential families with application to exponential random graph models.
\textit{Electron. J. Stat.}
\textbf{3}, 446-484.

\bibitem{Robins}
Robins, G., Pattison, P., Kalish, Y. and D. Lusher. (2007).
An introduction to exponential random graph $(p^{\ast})$ models for social networks.
\textit{Soc. Netw.}
\textbf{29}, 173-191.

\bibitem{Snijders}
Snijders, T. A. B., Pattison, P., Robins, G. L. and M. Handcock. (2006).
New specifications for exponential random graph models.
\textit{Sociol. Methodol.}
\textbf{36}, 99-153.

\bibitem{Duijn}
van Duijn, M. A. J., Snijders, T. A. B. and B. J. H. Zijlstra. (2004).
$p_{2}$: a random effects model with covariates for directed graphs.
\textit{Stat. Neerl.}
\textbf{58}, 234-254.

\bibitem{Wasserman}
Wasserman, S. and K. Faust. (2010).
\textit{Social Network Analysis: Methods and Applications (Structural Analysis in the Social Sciences)}.
Cambridge University Press, Cambridge.

\bibitem{Yin}
Yin, M. (2013).
Critical phenomena in exponential random graphs.
\textit{J. Stat. Phys.}
\textbf{153}, 1008-1021.

\bibitem{YinII}
Yin, M., Rinaldo, A. and S. Fadnavis. (2013).
Asymptotic quantization of exponential random graphs.
\textit{arXiv:} 1311.1738.

\bibitem{YinZhu}
Yin, M. and L. Zhu. (2014).
Asymptotics for sparse exponential random graph models.
\textit{arXiv:} 1411.4722.

\bibitem{Zhu}
Zhu, L. (2014).
Asymptotic structure of constrained exponential random graph models.
\textit{arXiv:} 1408.1536.

\end{thebibliography}
\end{document}